%% file: root.tex
\documentclass[letterpaper, 10 pt, conference]{ieeeconf}  % Comment this line out if you need a4paper

\usepackage{color,array}

\usepackage{amssymb,amsfonts}

\usepackage{graphicx}
\usepackage{verbatim}   

\usepackage{algorithm}
\usepackage{algpseudocode}

\usepackage{mathtools}
\usepackage{siunitx}
\usepackage{tabularx}
\usepackage{pgfplots}
\usepackage{pgfplotstable}

\usepackage{multirow}
\usepgfplotslibrary{fillbetween}
\usepackage{booktabs}
\pgfplotsset{compat = newest}
\usetikzlibrary{arrows,positioning,shapes,intersections,external,patterns,calc,fit,decorations,decorations.markings} 
\usepackage{cite}
\input{mydefs.tex}

\IEEEoverridecommandlockouts                              % This command is only needed if 
\title{\LARGE \bf
Data-driven Bayesian Control of Port-Hamiltonian Systems
}

\author{Thomas Beckers% <-this % stops a space
\thanks{Thomas Beckers is with the Department of Computer Science, Vanderbilt University, Nashville, TN 37212, USA {\tt\small thomas.beckers@vanderbilt.edu}}%
}

\begin{document}

\maketitle
\thispagestyle{empty}
\pagestyle{empty}

%%%%%%%%%%%%%%%%%%%%%%%%%%%%%%%%%%%%%%%%%%%%%%%%%%%%%%%%%%%%%%%%%%%%%%%%%%%%%%%%
\begin{abstract}
Port-Hamiltonian theory is an established way to describe nonlinear physical systems widely used in various fields such as robotics, energy management, and mechanical engineering. This has led to considerable research interest in the control of Port-Hamiltonian systems, resulting in numerous model-based control techniques. However, the performance and stability of the closed-loop typically depend on the quality of the PH model, which is often difficult to obtain using first principles. We propose a Gaussian Processes (GP) based control approach for Port-Hamiltonian systems (GPC-PHS) by leveraging gathered data. The Bayesian characteristics of GPs enable the creation of a distribution encompassing all potential Hamiltonians instead of providing a singular point estimate. Using this uncertainty quantification, the proposed approach takes advantage of passivity-based robust control with interconnection and damping assignment to establish probabilistic stability guarantees. 
\end{abstract}

\section{Introduction}
The modeling and control of physical systems is a crucial task in a broad range of domains such as physics, engineering, applied mathematics, and medicine~\cite{derler2011modeling}. Applications range from model-based control of autonomous systems~\cite{brosilow2002techniques} over the perception of dynamical objects~\cite{yin2021modeling} to the understanding of complex chemical processes~\cite{ikonen2001advanced}. A large class of physical systems can be described by Port-Hamiltonian systems (PHS), see~\cite{van2014port}.

PHSs are a powerful mathematical framework for modeling and analyzing physical systems, widely used in control theory, robotics, mechanical engineering, and other related fields. The PHS approach describes physical systems as a set of interconnected subsystems, each representing a specific physical domain such as electronics or mechanics. The dynamics of each subsystem is described by Hamiltonian equations, while the interconnections between subsystems are described by power-based port variables. This unique approach allows for systematic analysis of complex physical systems, providing a clear and intuitive understanding of their behavior, and enabling the design of efficient control strategies. Thus, in recent years, the control of PHS has gained significant attention in research, see \cite{van2000l2}.  

The interconnection and damping assignment - passivity based control (IDA-PBC) technique, introduced in~\cite{ortega2002interconnection}, is a well-known PBC design method for the control of physical system,
where the desired closed-loop dynamics takes the form of a PHS. 
The IDA-PBC methodology involves two main steps. The first step is the interconnection assignment, which involves the design of interconnections between subsystems that ensure the stability of the overall system. 
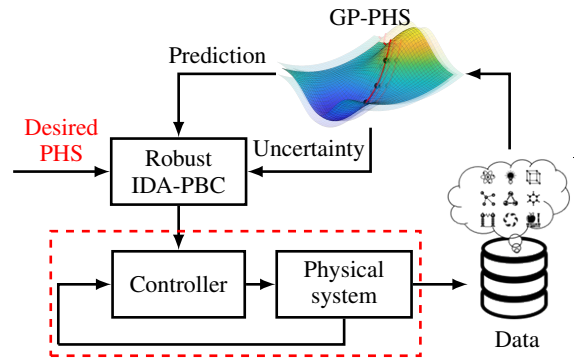
\begin{figure}[t]
\begin{center}
	\input{figure/bsb_intro.tex}
	\vspace{-0.2cm}\caption{Overview of the GPC-PHS approach. First, a GP-PHS model is used to learn the partially unknown dynamics of the physical systems. Then, the model enables the development of an IDA-PBC-based controller that is robustified against the model mismatch given by the uncertainty quantification of the GP-PHS.}\vspace{-0.7cm}
	\label{fig:gpcphs}
\end{center}
\end{figure}
The second step is damping assignment, which involves the design of damping terms to regulate the energy flow within the system and ensure its passivity. It has been successfully applied to a wide range of systems, including mechanical systems, power systems, and biological systems~\cite{gomez2004physical,ortega2004interconnection,9483212}.

However, this approach requires knowledge about the Hamiltonian, the interconnection, and damping matrix of the system to be controlled. These governing equations can be challenging to obtain due to highly nonlinear dynamics or the complexity of the physical system. Recently, we have introduced Gaussian process Port-Hamiltonian systems (GP-PHS) as Bayesian data-driven techniques for the identification of physical systems, see~\cite{9992733}. Gaussian processes (GPs) are a powerful and flexible machine learning tool that has gained significant attention in recent years. GPs provide a probabilistic framework for modeling complex and noisy data, enabling not only predictions but also uncertainty quantification. The goal of this paper is to combine the power of IDA-PBC based control approaches with data-driven Bayesian GP-PHS models to enable the control of physical systems with partially unknown dynamics.

In~\cite{nageshrao2015port}, the authors summarize data-driven control approaches for PHS. Adaptive control approaches for PHS are presented in~\cite{wang2007simultaneous,dirksz2010adaptive} and a reinforcement control for PHS in~\cite{sprangers2014reinforcement}. However, these approaches are limited to parametric uncertainties and lack of analysis tools, as well as a detailed study of the properties of the closed-loop system. In contrast, iterative learning control~\cite{fujimoto2003iterative} and repetitive control~\cite{fujimoto2004iterative} allow control of PHS without a priori knowledge of the system, but require repetitive executions. A controller that is robust to model uncertainties is introduced in~\cite{ryalat2018robust} but suffers from the trade-off between performance and robustness.

\textbf{Contribution:} In this paper, we propose a data-driven Bayesian control approach for physical systems with partially unknown dynamics. For this purpose, we employ Gaussian Process Port-Hamiltonian system models to learn the unknown dynamics of the system based on state measurements. Then, we introduce an IDA-PBC based control approach, which uses the mean prediction of the GP-PHS model to render the closed-loop dynamics to the desired PHS as close as possible. The uncertainty quantification of the GP-PHS model allows us to robustify the controller against the model mismatch such that the closed-loop has a stable desired equilibrium with high probability. Finally, we show that the closed-loop dynamics converges to the desired PHS for increasing amount of data.

The remainder of the paper is structured as follows. We introduce the idea of PHS and the problem setting in~\cref{sec:def}, followed by the proposed GPC-PHS control law in~\cref{sec:ctrl}. Finally, a simulation shows the benefits of GPC-PHS in~\cref{sec:sim}.

%%%%%%%%%%%%%%%%%%%%%%%%%%%%%%%%%%%%%%%%%%%%%%%%%%
%%%%%%%%%%%%%%%%%%%%%%%%%%%%%%%%%%%%%%%%%%%%%%%%%%
\section{Preliminaries}\label{sec:def}
In this section, we briefly describe the class of Port-Hamiltonian systems and introduce the problem setting. 
%%%%%%%%%%%%%%%%%%%%%%%%%%%%%%%%%%%%%%%%%%%%%%%%%%
\subsection{Port-Hamiltonian Systems}
Composing Hamiltonian systems with input/output ports leads to a Port-Hamiltonian system, which is a dynamical system with ports that specify the interactions of its components. The dynamics of a PHS is fully described by\footnote{Vectors~$\bm a$ and vector-valued functions~$\bm f(\cdot)$ are denoted with bold characters. Matrices are described with capital letters. $I_n$ is the $n$-dimensional identity matrix and $0_n$ the zero matrix. The expression~$A_{:,i}$ denotes the i-th column of $A$. For a positive semidefinite matrix $\Lambda$, $\|x - y\|_{\Lambda}^2 = (x - y)^\top \Lambda (x-y)$.  $\R_{>0}$ denotes the set of positive real numbers, while $\R_{\geq 0}$ is the set of non-negative real numbers. $\C^i$ denotes the class of $i$-th times differentiable functions. The operator $\nabla_\x$ with $\x\in\R^n$ denotes $[\frac{\partial}{\partial x_1},\ldots,\frac{\partial}{\partial x_n}]^\top$.}
\begin{align}
\begin{split}
\dx&=[J(\x)-R(\x)]\nabla_\x H( \x)+G(\x)\u\\
\y&=G(\x)^\top \nabla_\x H(\x),\label{for:pch}
\end{split}
\end{align}
with the state $\x(t)\in\R^n$ (also called energy variable) at time $t\in\R_{\geq 0}$, the total energy represented by a smooth function $H\colon\R^n\to\R$ called the  Hamiltonian, and the I/O ports $\u(t),\y(t)\in\R^m$.
\begin{rem}
We focus here on \emph{input-state-output} PHSs where there are no additional algebraic constraints on the
state variables, see~\cite{cervera2007interconnection}.
\end{rem}
The matrix $J\colon\R^n\to\R^{n\times n}$ is skew-symmetric and specifies the interconnection structure and the matrix $R\colon\R^n\to\R^{n\times n},\R=\R^\top\succeq 0$ specifies the dissipation in the system. Interaction with the environment is defined by the matrix $G\colon\R^n\to\R^{n\times m}$. The structure of the interconnection matrix~$J$ is typically derived from kinematic constraints in mechanical systems, Kirchhoff’s laws, power transformers, gyrators, etc.  Loosely speaking, the interconnection of the elements in the PHS is defined by $J$, whereas the Hamiltonian $H$ characterizes their dynamical behavior. The port variables~$\u$ and $\y$ are conjugate variables in the sense that their duality product defines the power flows exchanged with the environment of the system, for instance, currents and voltages in electrical circuits or forces and velocities in mechanical systems, see~\cite{van2000l2} for more information on PHS.
\begin{rem}
PHSs are a generalization of classical Hamiltonian systems but with the capability of including dissipation, input/output ports, and non-canonical coordinates. Thus, any Hamiltonian system can be represented by a PHS~\cref{for:pch} with
\begin{align*}
    J(\x)=\begin{bmatrix} 0 & I\\-I & 0\end{bmatrix},\,R=0,\,G=0
\end{align*}
\end{rem}
The class of PHS extends beyond the kinds of physical systems from which it was originally motivated. For example, many nonlinear systems with an asymptotically stable equilibrium point can be transformed into a PHS \cite{ortega2002interconnection}. 
%%%%%%%%%%%%%%%%%%%%%%%%%%%%%%%%%%%%%%%%%%%%%%%%%%
\subsection{Problem Setting}\label{sec:ps}
We consider the problem of designing a control law for a partially unknown physical system whose dynamics can be written in Port-Hamiltonian form~\cref{for:pch}. We assume that we have access to noisy observations $\tilde{\x}(t)\in\R^n$ of the system state $\x(t)\in\R^n$ whose evolution over time $t\in\R_{\geq 0}$ follows 
\begin{align}
\label{for:pchobs}
        \dx(t)=[J(\x)-R(\x)]\nabla_\x H( \x)+G(\x)\u(t)
\end{align} 
starting at $\x(0)\in\R^n$. The Hamiltonian $H\in\C^\infty$ is assumed to be \emph{completely unknown} due to unstructured uncertainties in the system typically imposed by nonlinear springs, physical coupling effects, or highly nonlinear electrical and magnetic fields. The parametric structures of the interconnection matrix $J\colon\R^n\to\R^{n\times n}$, dissipation matrix $R\colon\R^n\to\R^{n\times n}$ and I/O matrix $G\colon\R^n\to\R^{n\times m}$ are assumed to be known, but the parameters themselves might be unknown. Given a dataset of timestamps $\{t_i\}_{i=1}^N$ and noisy state observations with inputs, $\{\tilde \x(t_i),\bm{u}(t_i)\}_{i=1}^N$, our aim is to learn a control law $\bm{u}_c\colon\R_{\geq 0}\to\R^m$ that renders the system~\cref{for:pch} in the closed-loop to a desired PHS given by
\begin{align}\label{for:pchmodel}
\begin{split}
        \dx&=[J_d(\x)-R_d(\x)]\nabla_\x H_d( \x).\\
\end{split}
\end{align}
The dataset $\tilde \x(t_i)$ is assumed to be generated according to $\tilde \x(t_i) = \x(t_i) + \bm{\eta}$ where $\x(t)$ comes from the system~\cref{for:pchobs} with zero-mean Gaussian noise $\bm{\eta}\sim\mathcal{N}(\bm{0},\diag[\sigma_1,\ldots,\sigma_n])$. The variances $\sigma_1,\ldots,\sigma_n\in\R_{\geq 0}$ might be unknown.
%%%%%%%%%%%%%%%%%%%%%%%%%%%%%%%%%%%%%%%%%%%%%%%%%%
%%%%%%%%%%%%%%%%%%%%%%%%%%%%%%%%%%%%%%%%%%%%%%%%%%

%%%%%%%%%%%%%%%%%%%%%%%%%%%%%%%%%%%%%%%%%%%%%%%%%%
\section{Control of PHS}
\label{sec:ctrl}
The general idea of GPC-PHS is depicted in~\cref{fig:gpcphs}. First, we collect (noisy) state measurements from the physical system, which are then used to train a GP-PHS model. Due to its Bayesian nature, this model provides not only a prediction of the dynamics of the physical system but also uncertainty quantification. The mean prediction is used to design a control law based on IDA-PBC. The uncertainty prediction of the GP-PHS model is leveraged to make the control law robust against the model mismatch. This allows us to guarantee the desired closed-loop behavior under mild assumptions. We start with a brief review of GP-PHS followed by presenting the IDA-PBC control law and the probabilistic guarantees.
%%%%%%%%%%%%%%%%%%%%%%%%%%%%%%%%%%%%%%%%%%%%%%%%%%
\subsection{Gaussian Process Port-Hamiltonian System}\label{sec:GPIntro}
A GP-PHS, introduced in~\cite{9992733}, is a probabilistic model for learning partially unknown PHS based on state measurements. The main idea is to model the unknown Hamiltonian with a GP while treating the parametric uncertainties in $J,R$ and $G$ as hyperparameters, see~\cref{fig:gpphs}. A Gaussian process $\mathcal{GP}(m_{\mathrm{GP}}(\bm{x}), k(\bm{x},\bm{x}^\prime)$ is a stochastic process on some set $\mathcal{X} \subseteq \R^n$ where any finite collection of points $\bm{x}^1,\ldots,\bm{x}^L\in\X$ follows a multivariate Gaussian distribution
\begin{align*}
    \begin{bmatrix}
        f(\bm{x}^1)\\\vdots\\f(\bm{x}^L)
    \end{bmatrix}
    \!\sim\mathcal{N}\!\left(\!\begin{bmatrix}
        m(\bm{x}^1)\\\vdots\\m(\bm{x}^L)
    \end{bmatrix}\!,\!
    \begin{bmatrix}
        k(\bm{x}^1,\bm{x}^1) & \!\ldots\! & k(\bm{x}^1,\bm{x}^L)\\
        \vdots  & \!\ddots\! & \vdots\\ 
        k(\bm{x}^L,\bm{x}^1) & \!\ldots\! & k(\bm{x}^L,\bm{x}^L)
    \end{bmatrix}\!\right)
\end{align*}
with mean function $m_{\mathrm{GP}}: \R^n \rightarrow \R$, kernel function $k: \R^n \times \R^n \rightarrow \R$, and sample $f \sim \mathcal{GP}(m_{\mathrm{GP}}, k)$. By leveraging that GPs are closed under affine operations, the dynamics of a PHS \cref{for:pch} is integrated into the GP by
\begin{align}
    \dx&\sim \GP({\hat G}(\x\mid\bm{\varphi}_G)\u,k_{phs}(\x,\x^\prime)),\label{for:gpphs}
\end{align}
where the new kernel function $k_{phs}$ is given by
\begin{align*}
    k_{phs}(\x,\x^\prime)&=\sigma_f^2\hat{J}_R(\x\mid \bm{\varphi}_J,\bm{\varphi}_R)\Pi(\x,\x^\prime)\hat{J}_R^\top(\x^\prime\mid \bm{\varphi}_J,\bm{\varphi}_R)\notag\\
    \Pi_{i,j}(\x,\x^\prime) &= \frac{\partial }{\partial z_i \partial z_j}\exp(-\|\z- \z^\prime\|_{\Lambda}^2)\Big\vert_{\z=\x,\z^\prime=\x^\prime}
\end{align*}
with the Hessian $\Pi\colon\R^n\times\R^n\to\R^{n \times n}$ of the squared exponential kernel, see~\cite{rasmussen2006gaussian}. Thus, the dynamics~\cref{for:gpphs} describes a prior distribution over PHS. The matrices $J,R$ and $G$ of the PHS system~\cref{for:pchobs} are estimated by $\hat{J}_R(\x\mid \bm{\varphi}_J,\bm{\varphi}_R)=\hat{J}(\x\mid \bm{\varphi}_J)-\hat{R}(\x\mid \bm{\varphi}_R)$ and $\hat{G}(\bm{x}\mid \bm{\varphi}_G)$. The unknown set of parameters is described by $\bm{\varphi}_J\in\Phi_J\subseteq\R^{n_{\varphi_J}},{n_{\varphi_J}}\in\N$ for the estimated interconnection matrix $\hat{J}(x\vert\bm{\varphi}_J)\in\R^{n\times n} $, $\bm{\varphi}_R\in\Phi_R\subseteq\R^{n_{\varphi_R}},{n_{\varphi_R}}\in\N$ for the estimated dissipation matrix $\hat{R}(x\vert\bm{\varphi}_R)\in\R^{n\times n}$ and $\bm{\varphi}_G\in\Phi_G\subseteq\R^{n_{\varphi_G}},{n_{\varphi_G}}\in\N$ for the estimated I/O matrix $\hat{G}(x\vert\bm{\varphi}_G)\in\R^{n\times m}$. Together with the signal noise $\sigma_f\in\R_{>0}$, the lengthscales $\Lambda=\diag(l_1^2,\ldots,l_n^2)\in\R_{>0}^{n}$ of the kernel $k_{phs}$, the parameter vectors $\bm{\varphi}_J,\bm{\varphi}_R,\bm{\varphi}_G$ are treated as hyperparameters.

We start the training of the GP-PHS by using the collected dataset of timestamps $\{t_i\}_{i=1}^N$ and noisy state observations with inputs $\{\tilde \x(t_i),\bm{u}(t_i)\}_{i=1}^N$ of~\cref{for:pchobs} in a filter to create a dataset consisting of pairs of states $ X=[\x(t_1),\ldots,\x(t_N)]\in\R^{n\times N}$ and state derivatives $\dot{X}=[\dx(t_1),\ldots,\dx(t_N)]\in\R^{n\times N}$. 
Then, the unknown (hyper)parameters $\bm\varphi$ can be computed by minimization of the negative log marginal likelihood $-\log \prob(\dot{X}\vert \varphi,X)\sim\dot{X}_0^\top K_{phs}^{-1} \dot{X}_0+\log\vert K_{phs} \vert$, with the mean-adjusted output data  $\dot{X}_0=[[\dx(t_1)-\hat{G}\bm{u}(t_1)]^\top,\ldots,[\dx(t_{N})-\hat{G}\bm{u}(t_{N})]^\top]^\top$. Once the GP model is trained, we can compute the posterior distribution using the joint distribution with mean-adjusted output data $\dot{X}_0$ at a test states $\x^*\in\R^n$
\begin{align}
    \begin{bmatrix}\dot{X}_0\\ \dx \end{bmatrix}\!=\!\mathcal{N}\left(\bm{0},\begin{bmatrix}K_{phs} & k_{phs}(X,\x^*)\\k_{phs}(X,\x^*)^\top & k_{phs}(\x^*,\x^*)\end{bmatrix}\right)\notag.
\end{align}
Analogously to vanilla GP regression, the posterior distribution is then fully defined by the mean $\mu\left(\dx\!\mid\!\x^{*}, \D\right)$ and the variance $\var\left(\dx\!\mid\!\x^{*}, \D\right)$. For more detailed information on GP-PHS see~\cite{9992733}. 
%%%%%%%%%%%%%%%%%%%%%%%%%%%%%%%%%%%%%%%%%%%%%%%%%%%%%%%%%%%%
%%%%%%%%%%%%%%%%%%%%%%%%%%%%%%%%%%%%%%%%%%%%%%%%%%%%%%%%%%%%
\subsection{IDA-PBC for Bayesian Port-Hamiltonian systems}
In this section, we propose the data-driven based control approach for partially unknown PHS leveraging IDA-PBC and GP-PHS. IDA-PBC was introduced 
\begin{figure}[t]
\begin{center}
\vspace{0.2cm}
	\input{figure/bsb.tex}
	\vspace{-0.6cm}\caption{Block diagram of a Gaussian Process Port-Hamiltonian system.}\vspace{-0.6cm}
	\label{fig:gpphs}
\end{center}
\end{figure}
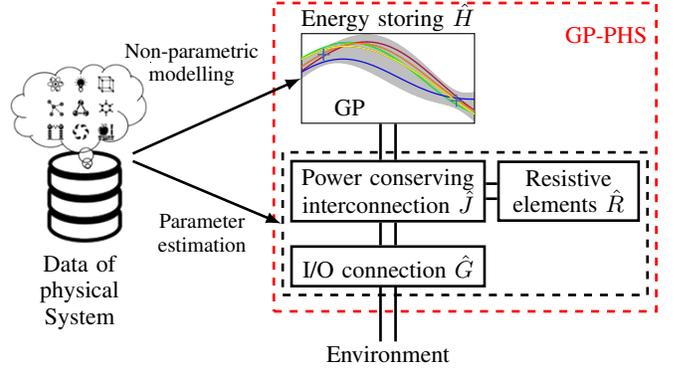
in~\cite{ortega1999energy,ortega2002interconnection} as a procedure to control physical systems described by PHS models such that the behavior of the closed-loop follows a desired PHS. Assume that there is a skew-symmetric matrix $J_d\colon\X\to\R^n$, a positive semidefinite diagonal matrix $R_d\colon\X\to\R^n$, and a function $H_d\colon\X\to\R$ on a compact set $\X\subseteq\R^n$ that satisfies the PDE
\begin{align}
    G^\perp(\x)[J(\x)-R(\x)]\nabla H\!=\!G^\perp(\x)[J_d(\x)-R_d(\x)]\nabla H_d,\label{for:pde}
\end{align}
where $G^\perp(\x)$ is a full-rank left annihilator of $G(\x)$.
\begin{defn}
A matrix function $G^\perp\colon\X\to\R^{(n-m)\times n}$ is a full-rank left annihilator of $G\colon\X\to\R^{n\times m}$ on $\X$ if $G(\x)^\perp G(\x)=\bm{0}$ and $\rank G(\x)^\perp=n-m$ for all $\x\in\X$. 
\end{defn}
\begin{assum}\label{ass:3}
The desired Hamiltonian $H_d$ is such that $\x_d=\arg\min H_d(\x)$, where $\x_d\in\X$ is the equilibrium to be stabilized.
\end{assum}
Then, under~\cref{ass:3}, the PCH system~\cref{{for:pch}} with control input $\bm{u}=\bm{\beta}(\x)$, where
\begin{align*}
    \bm{\beta}(\x)=&[G^\top(\x)G(\x)]^{-1}G^\top(\x)\notag\\
    &([J_d(\x)-R_d(\x)]\nabla H_d-[J(\x)-R(\x)]\nabla H)
\end{align*}
leads to the PHS system
\begin{align}
    \dx=J_d(\x)-R_d(\x)]\nabla H_d,\label{for:clPCH}
\end{align}
with a stable equilibrium at $\x_d$.
\begin{rem}
For mechanical systems, IDA-PBC control is equivalent to the theory of controlled Lagrangians~\cite{van2000l2}. Furthermore, for a review on the connection between optimal control and PBC, the reader is referred to \cite{vu2018connection}
\end{rem}
However, IDA-PBC requires full knowledge of the governing system equation~\cref{for:pchobs} which is typically hard to obtain, see~\cref{sec:ps}. In the following, we introduce an IDA-PBC control approach using the data-driven GP-PHS model instead of the unknown system dynamics. To overcome stability issues due to the model mismatch, the controller is robustified based on the uncertainty of the GP-PHS model. Before we introduce the main result, we impose the following assumption.
\begin{assum}\label{ass:1}
   There exists a constant $p\in(0,1)$ such that $P(|\mu\left(\dot{x}_i\!\mid\!\x, \D\right)-(J(\x)-R(\x))\nabla H(\x)|\leq \beta_i \var\left(\dot{x}_i\!\mid\!\x, \D\right),\forall i\in\{1,\ldots,n\},x\in\X\})\geq 1-p$
\end{assum}
\Cref{ass:1} is a standard assumption in GP learning which limits the class of unknown functions to be learned to the class of functions that the GP model can represent. See, for example,~\cite{srinivas2012information} for further information. 
\begin{thm}\label{thm:1}
Let~\cref{for:gpphs} be a GP-PHS model of the physical system~\cref{for:pchobs} based on the dataset $\D$ and let~\cref{ass:1,ass:3} be fulfilled. Assume that there is a skew-symmetric matrix $J_d\colon\X\to\R^n$, a positive semidefinite diagonal matrix $R_d\colon\X\to\R^n$, and a function $H_d\colon\X\to\R$ on a compact set $\X\subseteq\R^n$ that satisfy
\begin{align}
    \hat{G}^\perp(\x)\mu\left(\dx\!\mid\!\x, \D\right) =\hat{G}^\perp(\x)[J_d(\x)-R_d(\x)]\nabla H_d\label{for:spde}
\end{align}
where $H_d,R_d$ are designed such that
\begin{align}
    [\nabla H_d]^\top \bm{\eta}(\x)\leq[\nabla H_d]^\top R_d(\x) \nabla H_d\label{for:ineq}
    \end{align}
for all $\{\bm\eta(\x)\in\R^n| |\eta_i(\x)|\leq \beta_i  \var\left(\dot{x}_i\mid\x, \D\right),\,\forall i\in\{1,\ldots,n\},\x\in\X\}$. Then, the control input
\begin{align}
    \bm{u}(\x)=&[\hat{G}^\top(\x)\hat{G}(\x)]^{-1}\hat{G}^\top(\x)\notag\\
    &([J_d(\x)-R_d(\x)]\nabla H_d-\mu\left(\dx\!\mid\!\x, \D\right)\label{for:ctrl}
\end{align}
for the PHS~\cref{for:pchobs} leads to a closed-loop system with a stable equilibrium $\x_d$ on $X$ with probability $(1-p)$.
\end{thm}
\begin{proof}
We can rewrite the partially unknown PHS~\cref{for:pchobs}
\begin{align}
    \dx=[J(\x)-R(\x)]\nabla H+G(\x)\bm{u}\label{prf:phs}
\end{align}
in terms of the GP-PHS model as perturbed system
\begin{align}
    \dx=\mu\left(\dx\!\mid\!\x, \D\right)+\hat G(\x)\bm{u}+\bm{\eta}(\x),\label{prf:uphs}
\end{align}
with perturbation $\bm{\eta}$. Under~\cref{ass:1}, the GP-PHS model allows us to upper-bound the uncertainty by $P(|\eta_i(\x)|\leq \beta_i  \var\left(\dot{x}_i\mid\x, \D\right),\,\forall i\in\{1,\ldots,n\})\geq (1-p)$. As consequence, there exists a $\bm{\eta}$ such that \cref{prf:phs} equals~\cref{prf:uphs}. Following the idea of IDA-PBC, the control input~\cref{for:ctrl} applied to~\cref{prf:uphs} leads to
\begin{align}
    \dx=[J_d(\x)-R_d(\x)]\nabla H_d+\bm\eta(\x),\label{prf:duphs}
\end{align}
if the PDE~\cref{for:spde} holds. Finally, we choose $H_d$ as a Lyapunov-like function to prove that $\x_d$ is a stable equilibrium of the perturbed PHS~\cref{prf:duphs}. The evolution of $H_d$ is given by
\begin{align*}
    \dot{H}_d&=[\nabla H_d]^\top (J_d(\x)-R_d(\x)) \nabla H_d+[\nabla H_d]^\top \bm{\eta}(\x)\\
    &=-[\nabla H_d]^\top R_d(\x) \nabla H_d+[\nabla H_d]^\top \bm{\eta}(\x).
\end{align*}
With~\cref{for:ineq}, that leads to $P(\dot{H}_d\leq 0)\geq 1-p$, which concludes the proof.
\end{proof}
\begin{cor}
The equilibrium $\x_d$ will be asymptotically stable with probability $1-p$ if, in addition to~\cref{thm:1}, $\x_d$ is an isolated
minimum of $H_d$ and the largest invariant set under the closed-loop dynamics \cref{for:clPCH} contained in
\begin{align*}
    \{\x\in\X\vert[\nabla H_d]^\top R_d(\x)\nabla H_d=0\}
\end{align*}
equals the desired equilibrium $\{\x_d\}$. 
\end{cor}
\begin{proof}
    That is a direct consequence of Proposition 1 in~\cite{ortega2002interconnection}.
\end{proof}
An estimate of its domain of attraction
is given by the largest bounded level set $\{ \x\in\X\vert H_d(\x)\leq c\}$. Higher noise in the state measurements would increase the bound in~\cref{ass:1} so that the controller needs to be "more robust" by designing $H_d$ and $R_d$ to satisfy condition~\cref{for:ineq}. So far,~\cref{thm:1} states that the control law~\cref{for:ctrl} ensures that $\x_d$ is a stable equilibrium of the closed-loop system. In addition, closed-loop behavior follows the desired PHS \cref{for:pchmodel} affected by a perturbation $\bm{\eta}$ that depends on the uncertainty of the GP-PHS model. Next, we show that for an increasing amount of training data, the closed-loop behavior converges to the desired PHS.
\begin{lem}\label{lem:1}
In addition to~\cref{thm:1}, if the number of data points $N\to\infty$, where $\x(t_1)\neq \x(t_2)\neq \ldots \neq \x(t_N)$ with $\x(t_i)\in\X$, then the control law~\cref{for:ctrl} applied to~\cref{prf:uphs} leads to the desired closed-loop dynamics $\dx=[J_d(\x)-R_d(\x)]\nabla H_d$ with probability $(1-p)$.
\end{lem}
\begin{proof}
    Since $\var\left(\dot{x}_i\mid\x, \D\right)=0$ for all $x\in\X$ if the number of data points tends to infinity, see~\cite{umlauft:TAC2020}, $P(\Vert\eta(\x)\Vert=0,\forall \x\in\X)\geq 1-p$. As a consequence, the closed-loop system is rendered to $\dx=[J_d(\x)-R_d(\x)]\nabla H_d$.
\end{proof}
Although analytic solutions for~\cref{for:pde} can be obtained for
some classes of PHS, solving the IDA-PBC PDE is nontrivial in general.
In~\cite{ortega2004interconnection}, the authors propose three different ways to solve the matching equation~\cref{for:pde}.

\textit{Non-Parameterized IDA:} Here, the desired
interconnection $J_d$ and dissipation $R_d$
matrices are fixed as well as $G^\perp$.
This leads to a PDE whose solutions define the
admissible energy functions $H_d$ for the given
interconnection and damping matrices. Then, we select the solution that satisfies \cref{ass:3}.

\textit{Algebraic IDA:} At the other extreme, the desired energy function is fixed which yields~\cref{for:pde} to become an algebraic equation in $J_d$, $R_d$ and $G^\perp$.

\textit{Parameterized IDA:} For some physical systems, it
is desirable to restrict the desired energy function to a
certain class. By fixing the
structure of the energy function, the matching function is transformed into a new PDE with constraints on the interconnection and damping
matrices.

 We refer interested readers to~\cite{nageshrao2015port} for a review on methods to solve the IDA-PBC PDEs.
 \begin{rem}
    Despite "classical" techniques to solve the matching equation, physics-informed neural networks have shown promising results for finding solutions of the IDA-PBC PDE, for example, see~\cite{plaza2022total}.
\end{rem}
%%%%%%%%%%%%%%%%%%%%%%%%%%%%%%%%%%%%%%%%%%%%%%%%%%%%%%%%%%%%
%%%%%%%%%%%%%%%%%%%%%%%%%%%%%%%%%%%%%%%%%%%%%%%%%%%%%%%%%%%%
\section{Evaluation}
\label{sec:sim}
We consider the problem of designing a control law for a nonlinear electrostatic microactuator, as depicted in~\cref{fig:micro}. The system's equations in PHS form are given by
\begin{align}\label{for:sim}
        \dx&=\underbrace{\begin{bmatrix}
            0 & 1 & 0\\ -1 & -b & 0\\0 & 0 & -\frac{1}{r}
        \end{bmatrix}}_{J(x)-R(x)}\frac{\partial H}{\partial\x}( \x)+\underbrace{\begin{bmatrix}
            0\\0\\\frac{1}{r}
        \end{bmatrix}}_{G(x)}u\\
        H(\x)&=\frac{1}{4}10(x_1-x^*)^4+\frac{1}{2m}x_2^2+\frac{x_1}{2A\epsilon}x_3^2.\notag
\end{align}
\begin{figure}[t]
\begin{center}
\vspace{0.2cm}
	\includegraphics[width=0.8\columnwidth]{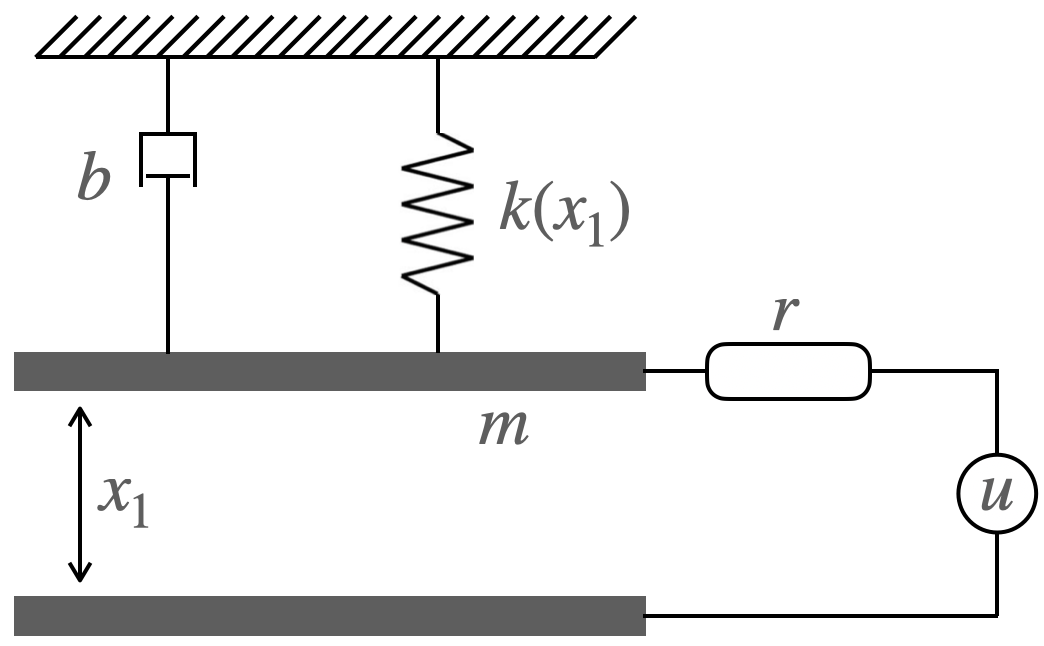}
	\caption{Electrostatic micro-actuator with nonlinear spring $k$.}\vspace{-0.6cm}
	\label{fig:micro}
\end{center}
\end{figure}
with the air gap $x_1$, the momentum $x_2$ and the charge of the device $x_3$, inspired by~\cite{maithripala2003nonlinear}.
The parameters are the plate area $A=1$, the mass of the plate $m=1$ and the permittivity in the gap $\epsilon=1$. The steady state of the air gap is $x_1^*=1$. We assume a linear damping with positive constant $b=0.5$ and a non-linear, position-dependent stiffness $k(x_1)$ for the spring. The input resistance is $r=1$ and $u$ represents the input voltage, which is the control input of the system. As stated in the problem setting, we assume that the Hamiltonian $H$ and the damping constant $b$ are \textit{unknown to us}. The goal is to stabilize the system at an equilibrium point where the air gap $x_1=0.5$.

First, we collect a set of training data for the GP-PHS model. For this purpose, we use a sinusoidal input signal, i.e., $u(t)=\sin(t)$ as excitation of the microactuator system~\cref{for:sim}. The system is initialized with $\x(0)=[0,0,1]^\top$ and 300 data pairs $\{t_i,\x(t_i)\}$ are recorded between $0\si{\milli\second}$ and $20\si{\milli\second}$ with constant time spacing, see~\cref{fig:training}. The data is corrupted by zero-mean Gaussian noise with variance $\sigma^2=0.001$. 
With this dataset, a GP-PHS model is trained according to~\cite[Algorithm 1]{9992733}. Optimization of the hyperparameters results in an estimated damping constant $\hat{b}=0.498$. A comparison between an actual system trajectory and the mean prediction of the GP-PHS model is visualized in~\cref{fig:test}. Note the performance of the GP-PHS model on an a-priori unseen state-space even though the model was only trained
\begin{figure}[b]
\begin{center}
	\input{figure/training.tex}
	\caption{The training data for the GP-PHS model.}\vspace{-0.2cm}
	\label{fig:training}
\end{center}
\end{figure}
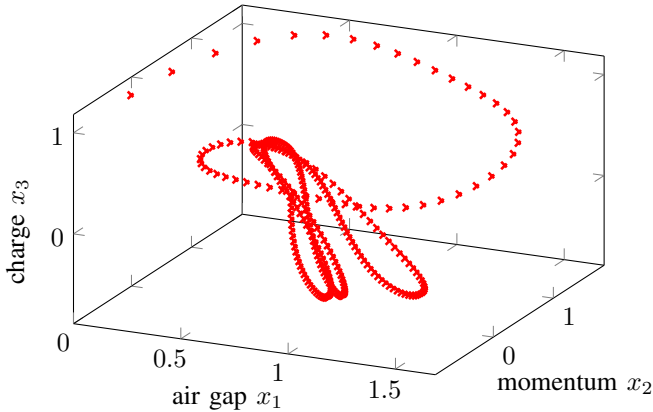
on a \textit{single} trajectory of the system.
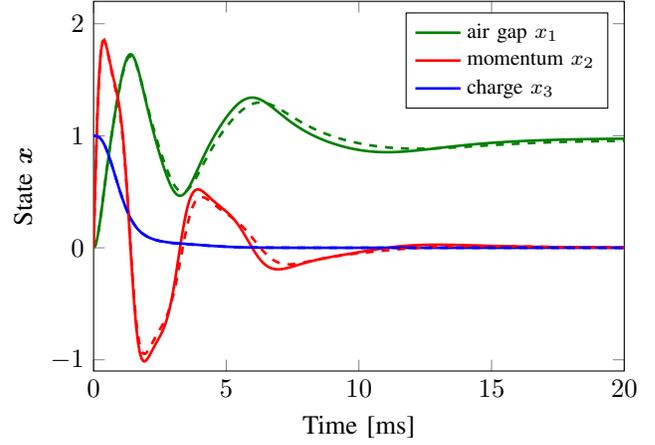
\begin{figure}[t]
\begin{center}
\vspace{0.2cm}
	\input{figure/test.tex}
	\caption{Comparison between the actual system (solid) and the mean predicition of the GP-PHS model (dashed).}\vspace{-0.6cm}
	\label{fig:test}
\end{center}
\end{figure}
Next, we use~\cref{thm:1} to derive a control law which stabilizes the system at a desired equilibrium point $\x_d$. For solving the PDE~\cref{for:spde}, the full-rank left annihilator of $G(\x)$ is determined to
\begin{align*}
    G^\perp(\x)=\begin{bmatrix}
        1 & 0 & 0\\0 & 1 & 0
    \end{bmatrix}.
\end{align*}
We follow the idea of non-parametric IDA by fixing the desired interconnection and damping matrix to
\begin{align*}
    J_d(x)-R_d(x)=\begin{bmatrix}
            0 & 1 & 0\\ -1 & -\hat{b} & 0\\0 & 0 & -\frac{1}{r_{d}}
        \end{bmatrix}
\end{align*}
 with $r_d=2/3$. As candidate for the desired Hamiltonian $H_d$, we use the posterior mean prediction for the Hamiltonian of the GP-PHS model and, inspired by~\cite{nageshrao2015port}, an additional term to adjust the equilibrium to a desired state. Thus, the desired Hamiltonian is given by
 \begin{align*}
     H_d(x)=\mu(\hat{H}\mid \x,\D)+\zeta(x_3),
 \end{align*}
 where the function $\zeta(x_3)=(x_3-c)^2$ with $c\in\R$ is determined, such that the desired equilibrium $\x_d=[0.5,c_2,c_3]$ with some $c_2,c_3\in\R$ is a minimum of $H_d(x)$ on $\X=[-2,2]^3$. For this desired Hamiltonian and $\beta_{1}=\beta_{2}=\beta_{3}=2$, we validate the PDE~\cref{for:spde} by point evaluations over a discretized state space $\X$. The top plot of \cref{fig:cl} shows the closed-loop dynamics with the proposed control law~\cref{for:ctrl}. As there is uncertainty in the GP-PHS model, the closed-loop trajectory (solid) deviates slightly from the desired PHS dynamics (dashed) $\dx=(J_d(x)-R_d(x))\nabla H_d(x)$.
 
However, according to~\cref{thm:1}, the desired equilibrium remains stable. The bottom plot supports this result as the desired Hamiltonian evaluated on the state of the closed-loop system decreases over time, as claimed by~\cref{for:ineq}. Finally, we compare the performance of the closed-loop over the number of data points.~\Cref{fig:moredata} shows that for an increasing amount of data, the closed-loop dynamics converges to the dynamics defined by the desired PHS.
\begin{figure}[ht]
\begin{center}
\vspace{0.2cm}
	\input{figure/cl.tex}
	\caption{Top: Closed-loop system with the proposed control law. As there is uncertainty in the GP-PHS model, the closed-loop dynamics (solid) slightly deviates from the desired PHS (dashed). However, in line with~\cref{thm:1}, the closed-loop converges to the desired equilibrium (green dot). Bottom: The desired storage function $H_d$ is decreasing over time.}\vspace{-0.2cm}
	\label{fig:cl}
\end{center}
\end{figure}
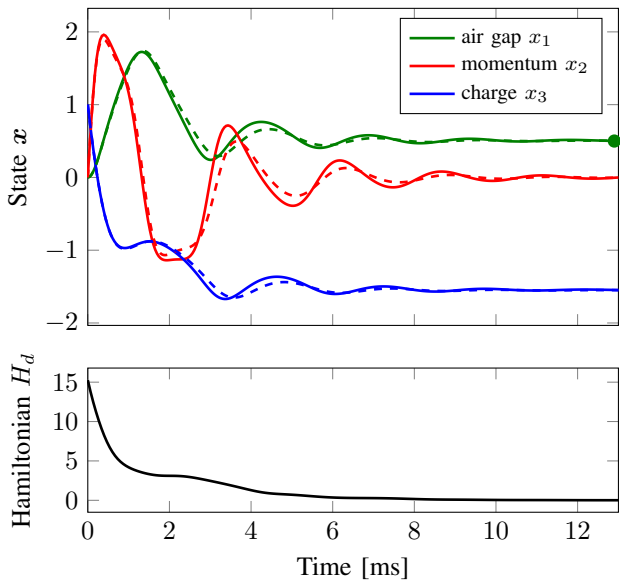
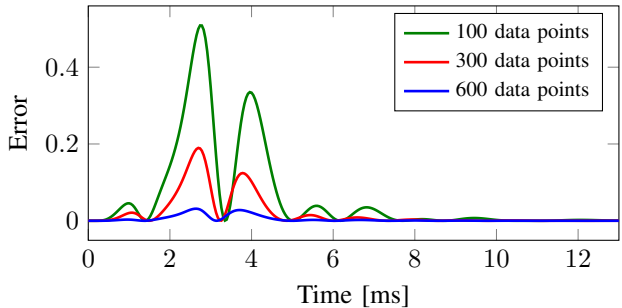
\begin{figure}[ht]
\begin{center}
	\input{figure/moredata.tex}
	\vspace{-0.1cm}\caption{The mean square error between the desired trajectory and the actual closed-loop trajectory. With more data, the closed-loop follows the desired PHS, as stated in~\cref{lem:1}.}\vspace{-0.8cm}
	\label{fig:moredata}
\end{center}
\end{figure}
\section*{Conclusion}
In this paper, we present a Bayesian data-driven control approach (GPC-PHS) for partially unknown physical systems. A GP-PHS model uses recorded data of the physical system to learn the unknown dynamics. Then, we propose an IDS-PBC based control law which leverages the uncertainty quantification of the GP-PHS model to robustify the controller against the model mismatch. As consequence, the stability of a desired equilibrium can be guaranteed with high probability. 
With an increase in data, the closed-loop dynamics converges to the desired PHS. A simulation shows the effectiveness of the proposed control approach. Future work will address the tracking control problem with time-dependent $\x_d$.

\bibliographystyle{IEEEtran}
\bibliography{root}

\end{document}

%% file: mydefs.tex
\newtheorem{defn}{Definition}
\newtheorem{rem}{Remark}
\newtheorem{thm}{Theorem}
\newtheorem{lem}{Lemma}
\newtheorem{cor}{Corollary}

\newtheorem{assum}{Assumption}

\newcommand\tran{\mkern-2mu\raise1.25ex\hbox{$\scriptscriptstyle\top\hspace{0.5mm}$}\mkern-3.5mu}
\newcommand{\R}{\mathbb{R}}

\newcommand{\N}{\mathbb{N}}
\newcommand{\C}{\mathcal{C}}

\newcommand{\D}{\mathcal{D}}
\newcommand{\X}{\mathcal{X}}

\newcommand{\bm}[1]{{\boldsymbol{#1}}}

\DeclareMathOperator{\diag}{diag}
\DeclareMathOperator{\var}{var}

\DeclareMathOperator{\rank}{rank}

\DeclareMathOperator{\prob}{p}
\newcommand{\GP}{\mathcal{GP}}
\newcommand{\z}{\bm z}

\newcommand{\x}{\bm x}

\newcommand{\m}{\bm m}
\newcommand{\dx}{\dot{\bm x}}

\newcommand{\f}{\bm{f}}

\renewcommand{\u}{\bm{u}}
\newcommand{\y}{\bm{y}}

\usepackage[noabbrev]{cleveref} 
\crefname{rem}{Remark}{Remarks}
\crefname{exam}{Example}{Examples}
\crefname{assum}{Assumption}{Assumptions}
\crefname{prop}{Proposition}{Propositions}
\crefname{propy}{Property}{Properties}
\crefname{cor}{Corollary}{Corollaries}
\crefname{lem}{Lemma}{Lemmas}
\crefname{section}{Section}{Sections}
\crefname{thm}{Theorem}{Theorems}
\crefname{alg}{Algorithm}{Algorithms}
\crefname{defn}{Definition}{Definitions}
\crefname{figure}{Fig.}{Fig.}
\Crefname{figure}{Figure}{Figures}
\crefname{equation}{}{}

%% file: figure/bsb_intro.tex
\tikzsetnextfilename{bsb_intro}
\begin{tikzpicture}[auto, node distance=2cm,>=latex]
	\tikzstyle{block} = [draw, fill=white, rectangle,  line width=1pt, 
    minimum height=2.5em, minimum width=5em, font=\small,align=center, inner sep=3pt]
	\tikzstyle{input} = [coordinate]

    % We start by placing the blocks
    \node [block,font=\small,align=center,text width=4.5em] (system) {Physical\\system};
     \node [right of=system, transform canvas={yshift=0.6cm,xshift=-0.3cm},inner sep=0pt, node distance=2.5cm,label={[label distance=-0.6cm,xshift=-0.2cm,font=\small,align=center,text width=5em]-90:Data}] (data) {\includegraphics[width=1.7cm]{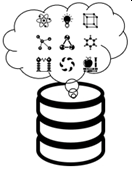}};
    \node[block, left of=system, node distance=2.2cm] (controller) {Controller};
    \node[block, above of=controller, node distance=1.5cm,align=center,text width=4.5em] (robust) {Robust\\IDA-PBC};
    \node[inner sep=0pt, xshift=1em, above of=system, node distance=2.8cm,label={[label distance=-0.2cm,font=\small,align=center,text width=5.2em]90:GP-PHS}] (gpphs) {\includegraphics[width=2.5cm]{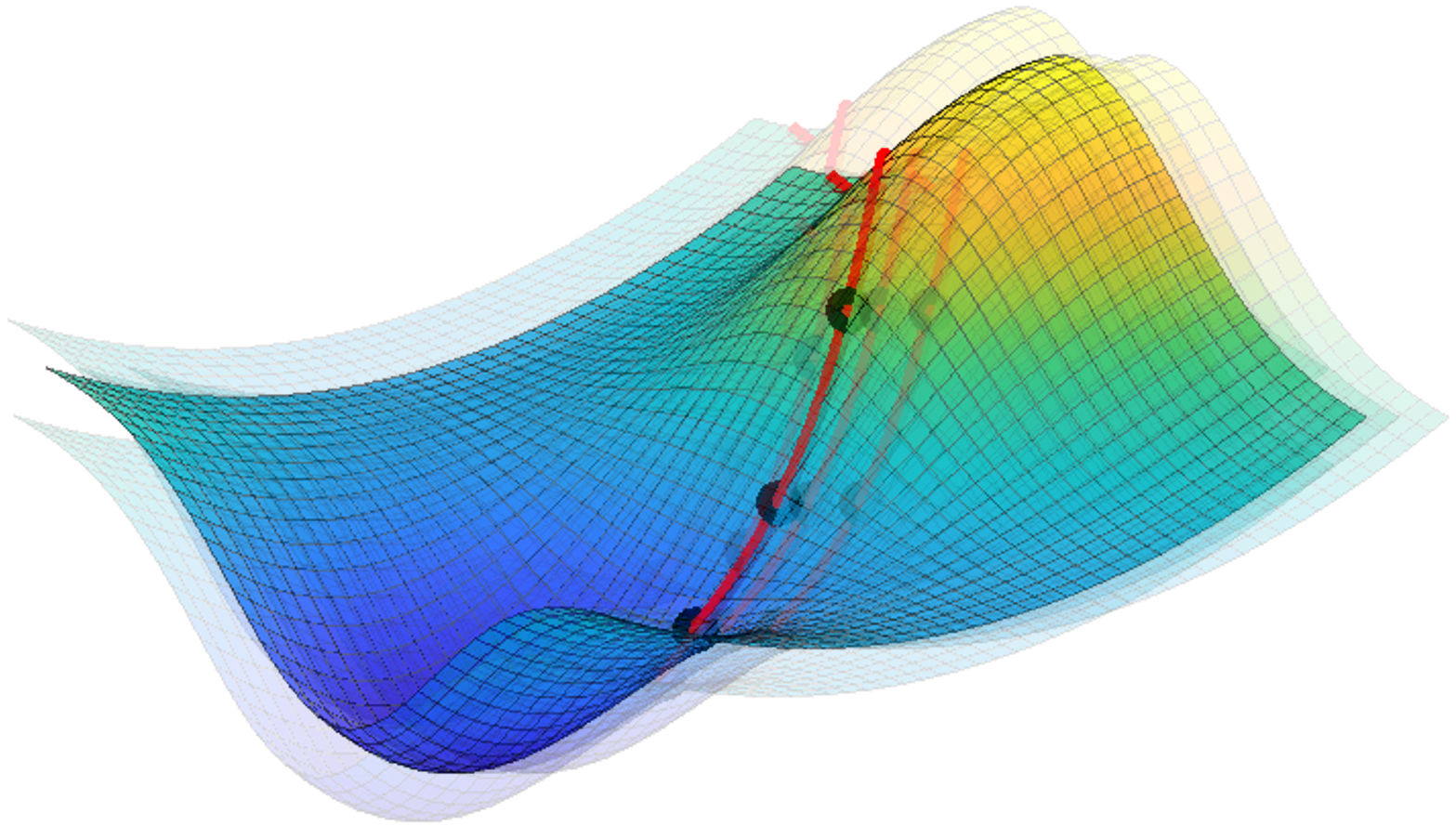}};
    \node[coordinate, left of=robust, node distance=2.2cm,font=\small,align=center,text width=5.2em] (des) {};
    \node[coordinate, left of=controller, node distance=1.6cm,yshift=-0.8cm] (helper) {};
    \draw [->,line width=1pt]  (des)  -- node [above,pos=0.5,font=\small,align=center,text width=4.5em] {\color{red}Desired\\PHS}  (robust.west);
    \draw [->,line width=1pt]  (robust.south) --  (controller.north);
    \draw [->,line width=1pt]  (controller.east) --  (system.west);
    \draw [->,line width=1pt]  (system.east) --  (data.west);
    \draw [->,line width=1pt]  ([yshift=0.7cm,xshift=-0.3cm]data.north) |-  (gpphs.east);
    \draw [->,line width=1pt]  (gpphs.west) -| node [above,pos=0.3,font=\small,align=center,text width=4.5em] {Prediction} (robust.north);
    \draw [->,line width=1pt]  (gpphs.south) |- node [above,pos=0.75,font=\small,align=center,text width=4.5em] {Uncertainty} (robust.east);
    \draw [line width=1pt]  (system.south) |-  (helper);
    \draw [->,line width=1pt]  (helper) |-  (controller.west);

    \node[draw,dashed,red,line width=1pt,yshift=0.2cm,yshift=-0.15cm,inner xsep=0.1cm,inner ysep=0.2cm,fit=(controller) (system) (helper)] (box1) {};
\end{tikzpicture}

%% file: figure/bsb.tex
\tikzsetnextfilename{bsb}
\begin{tikzpicture}[auto, node distance=2cm,>=latex]
	\tikzstyle{block} = [draw, fill=white, rectangle,  line width=1pt, 
    minimum height=1.5em, minimum width=2em, font=\small,align=center, inner sep=3pt]
	\tikzstyle{input} = [coordinate]

    % We start by placing the blocks
    \node [block, text width=6.7em] (connection) {Power conserving interconnection $\hat J$};
    \node [draw,inner sep=0pt, above of=connection, node distance=1.5cm,label={[label distance=-0.1cm,font=\small,align=center]90:Energy storing $\hat H$},label={[label distance=-0.7cm,font=\small,align=center]220:GP},] (energy) {\includegraphics[width=2.3cm]{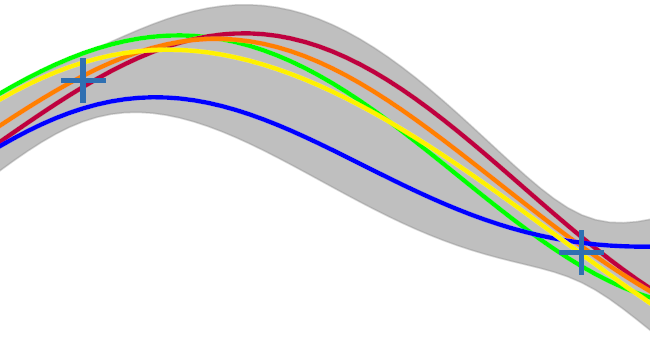}};
    \node [block, right of=connection, text width=4.7em, node distance=2.4cm] (res) {Resistive elements $\hat R$};
    \node [block, text width=6.7em, below of=connection, node distance=1cm] (iog) {I/O connection $\hat G$};
    \node [draw, input, below of=iog, node distance=1cm,label={[yshift=-0.4cm,font=\small,align=center]Environment}] (io) {};
    \node [inner sep=0pt, left of=connection, node distance=4.1cm,yshift=0.5cm,label={[label distance=0cm,font=\small,align=center,text width=4.7em]-90:Data of physical System}] (data) {\includegraphics[width=1.9cm]{figure/data.png}};
    \node[draw,dashed,line width=1pt,xshift=0cm,inner xsep=0.1cm,inner ysep=0.1cm,fit=(iog) (connection) (res)] (box) {};
    \node[draw,dashed,red,line width=1pt,yshift=0.2cm,yshift=-0.1cm,inner xsep=0.1cm,inner ysep=0.3cm,fit=(box) (energy),label={[label distance=-0.9cm,font=\small,align=center]49:{\textcolor{red}{GP-PHS}}}] (box1) {};
    
    \draw [-,line width=1pt] ([xshift=0.1cm]connection.north) --  ([xshift=0.1cm]energy.south);
    \draw [-,line width=1pt] ([xshift=-0.1cm]connection.north) --  ([xshift=-0.1cm]energy.south);
    \draw [-,line width=1pt] ([yshift=0.1cm]connection.east) -- ([yshift=0.1cm]res.west);
    \draw [-,line width=1pt] ([yshift=-0.1cm]connection.east) -- ([yshift=-0.1cm]res.west);
    \draw [-,line width=1pt] ([xshift=0.1cm]connection.south) -- ([xshift=0.1cm]iog.north);
    \draw [-,line width=1pt] ([xshift=-0.1cm]connection.south) -- ([xshift=-0.1cm]iog.north);
    \draw [-,line width=1pt] ([xshift=0.1cm]iog.south) -- ([xshift=0.1cm]io);
    \draw [-,line width=1pt] ([xshift=-0.1cm]iog.south) -- ([xshift=-0.1cm]io);
    \draw [->,line width=1pt] (data.east)++(-0.25cm,0) -- node[label={[xshift=-0.2cm,yshift=0cm,font=\footnotesize,align=center, text width=5.2em]Non-parametric modelling}] {} (energy.west);
    \draw [->,line width=1pt] (data.east)++(-0.25cm,-0.1cm) -- node[label={[xshift=-0.2cm,yshift=-1.2cm,font=\footnotesize,align=center, text width=3.5em]Parameter estimation}] {}
     (box.west);
\end{tikzpicture}

%% file: figure/training.tex
\begin{tikzpicture}
\begin{axis}[
  xlabel={air gap $x_1$},
  ylabel={momentum $x_2$},
  zlabel={charge $x_3$},
  xlabel shift = -10 pt,
  ylabel shift = -10 pt,
  legend pos=north west,
  width=\columnwidth,
  height=6.5cm,
  legend style={font=\footnotesize},
  legend cell align={left},
  %view={-30}{20},
  legend style={at={(0,1.01)},anchor=north west,/tikz/every even column/.append style={column sep=0.2cm}},
  legend columns=3]
\addplot3[only marks, color=red,dashed,line width=1pt,mark=x] table [x index=1,y index=2,z index=3]{data/data.dat};
\end{axis}
\end{tikzpicture} 

%% file: figure/test.tex
\tikzsetnextfilename{pogo_system}
\begin{tikzpicture}
\begin{axis}[
  name=plot1,
  xlabel={Time [ms]},
  ylabel={State $\x$},
  legend pos=north west,
  width=\columnwidth,
  height=6.5cm,
  ymin=-1.1,
  ymax=2.2,
  xmin=0,
  xmax=20,
  legend style={font=\footnotesize},
  legend cell align={left},
  legend pos=north east]
\addplot[color=green!50!black,line width=1pt,no marks] table [x index=0,y index=1]{data/test.dat};
\addplot[color=red,line width=1pt,no marks] table [x index=0,y index=2]{data/test.dat};
\addplot[color=blue,line width=1pt,no marks] table [x index=0,y index=3]{data/test.dat};
\addplot[color=green!50!black,dashed,line width=1pt,no marks] table [x index=0,y index=4]{data/test.dat};
\addplot[color=red,dashed,line width=1pt,no marks] table [x index=0,y index=5]{data/test.dat};
\addplot[color=blue,dashed,line width=1pt,no marks] table [x index=0,y index=6]{data/test.dat};
\legend{air gap $x_1$,momentum $x_2$, charge $x_3$};
\end{axis}
\end{tikzpicture} 

%% file: figure/cl.tex
\tikzsetnextfilename{cl0}
\begin{tikzpicture}
\begin{axis}[
  name=plot1,
  ylabel={State $\x$},
  legend pos=north west,
  width=\columnwidth,
  height=5.8cm,
  xmin=0,
  xmax=13,
  xticklabels={},
  legend style={font=\footnotesize},
  legend cell align={left},
  legend pos=north east]
\addplot[color=green!50!black,line width=1pt,no marks] table [x index=0,y index=4]{data/control.dat};
\addplot[color=red,line width=1pt,no marks] table [x index=0,y index=5]{data/control.dat};
\addplot[color=blue,line width=1pt,no marks] table [x index=0,y index=6]{data/control.dat};
\addplot[color=green!50!black,dashed,line width=1pt,no marks] table [x index=0,y index=1]{data/control.dat};
\addplot[color=red,dashed,line width=1pt,no marks] table [x index=0,y index=2]{data/control.dat};
\addplot[color=blue,dashed,line width=1pt,no marks] table [x index=0,y index=3]{data/control.dat};
\addplot[color=green!50!black,line width=1pt,mark=*] coordinates {(12.9,0.5)};
\legend{air gap $x_1$,momentum $x_2$, charge $x_3$};
\end{axis}
\begin{axis}[
  name=plot2,
   at=(plot1.below south east), anchor=above north east,
  xlabel={Time [ms]},
  ylabel={Hamiltonian $H_d$},
  legend pos=north west,
  width=\columnwidth,
  height=3.5cm,
  xmin=0,
  xmax=13,
  legend style={font=\footnotesize},
  legend cell align={left},
  legend pos=north east]
\addplot[color=black,line width=1pt,no marks] table [x index=0,y index=1]{data/Lyap.dat};
\end{axis}
\end{tikzpicture} 

%% file: figure/moredata.tex
\tikzsetnextfilename{error}
\begin{tikzpicture}
\begin{axis}[
  name=plot1,
  xlabel={Time [ms]},
  ylabel={Error},
  legend pos=north west,
  width=\columnwidth,
  height=4.7cm,
  xmin=0,
  xmax=13,
  legend style={font=\footnotesize},
  legend cell align={left},
  legend pos=north east]
\addplot[color=green!50!black,line width=1pt,no marks] table [x index=0,y index=1]{data/error.dat};
\addplot[color=red,line width=1pt,no marks] table [x index=0,y index=2]{data/error.dat};
\addplot[color=blue,line width=1pt,no marks] table [x index=0,y index=3]{data/error.dat};
\legend{100 data points, 300 data points, 600 data points};
\end{axis}
\end{tikzpicture} 